\documentclass[letterpaper,11pt]{article}

\usepackage{amsthm}
\usepackage{amsfonts}
\usepackage{amsmath}
\usepackage{graphicx}
\usepackage{epstopdf}
\usepackage{lineno}
\usepackage{setspace}
\usepackage{verbatim}

\newtheorem{theorem}{Theorem}
\newtheorem*{theorem*}{Theorem}
\newtheorem{conjecture}[theorem]{Conjecture}
\newtheorem{corollary}[theorem]{Corollary}

\newtheorem{lemma}[theorem]{Lemma}

\newcommand{\E}{\mathbf{E}}
\renewcommand{\Pr}{\mathbf{Pr}}

\newcommand{\fc}{\mathcal C}

\newcommand{\tN}{\tilde{N}}

\title{A short proof that $\chi$ can be bounded $\epsilon$ away from $\Delta+1$ towards $\omega$}
\author{Andrew D.\ King \and Bruce A.\ Reed}

\newenvironment{enumerate*}{
\begin{enumerate}
  \setlength{\itemsep}{5pt}
  \setlength{\parskip}{0pt}
  \setlength{\parsep}{0pt}
}{\end{enumerate}}

\newenvironment{itemize*}{
\begin{itemize}
  \setlength{\itemsep}{5pt}
  \setlength{\parskip}{0pt}
  \setlength{\parsep}{0pt}
}{\end{itemize}}

\newenvironment{proof*}[1][\proofname]{%
\proof
}{\endproof}

\begin{document}

\maketitle

\begin{abstract}
In 1998 the second author proved that there is an $\epsilon>0$ such that every graph satisfies $\chi \leq  \lceil (1-\epsilon)(\Delta+1)+\epsilon\omega\rceil$.  The first author recently proved that any graph satisfying $\omega > \frac 23(\Delta+1)$ contains a stable set intersecting every maximum clique.  In this note we exploit the latter result to give a much shorter, simpler proof of the former.  We include, as a certificate of simplicity, an appendix that proves all intermediate results with the exception of Hall's Theorem, Brooks' Theorem, the Lov\'asz Local Lemma, and Talagrand's Inequality.
\end{abstract}

\section{Introduction}

Much work has been done towards bounding the chromatic number $\chi$ of a graph in terms of the clique number $\omega$ and the maximum size of a closed neighbourhood $\Delta+1$, which are trivial lower and upper bounds on the chromatic number, respectively.  Recently, much of this work has been done in pursuit of a conjecture of Reed, who proposed that the average of the two should be an upper bound for $\chi$, modulo a round-up:

\begin{conjecture}[\cite{reed98}]
Every graph satisfies $\chi \leq \lceil \tfrac 12(\Delta+1+\omega)\rceil$.
\end{conjecture}

This conjecture has been proven for some restricted classes of graphs \cite{aravindks11, kingthesis, kingr08, kingrv07, rabern08}, sometimes in the form of a stronger local conjecture posed by King \cite{chudnovskykps11, kingthesis}; both forms are known to hold in the fractional relaxation \cite{kingthesis, molloyrbook}.

For general graphs, we only know that we can bound the chromatic number by some nontrivial convex combination of $\omega$ and $\Delta+1$:

\begin{theorem}[\cite{reed98}]\label{thm:main}
There exists an $\epsilon> 0$ such that every graph satisfies
$$\chi \leq \lceil (1-\epsilon)(\Delta+1) +\epsilon\omega\rceil.$$
\end{theorem}

The original proof of this theorem is quite long and complicated.  In this note we give a much shorter, simpler proof that exploits the following new existence condition for a stable set hitting every maximum clique, the proof of which from first principles is itself short and fairly simple:

\begin{theorem}[\cite{king11}]\label{thm:king}
Every graph satisfying $\omega > \frac 23(\Delta+1)$ contains a stable set hitting every maximum clique.
\end{theorem}

This result is a strengthening of a result of Rabern \cite{rabern11}, which could be used to similar effect.

\section{A proof sketch}

We sketch the proof here, prove the necessary lemmas in the following two sections, then finally prove the theorem more formally.

Suppose $G$ is a minimum counterexample to Theorem \ref{thm:main} for some fixed $\epsilon$.  Applying Theorem \ref{thm:king} and Brooks' Theorem tells us that $G$ satisfies $\omega \leq \frac 23(\Delta+1)$ and $\Delta > \frac{1}{\epsilon}$.  Our proof then considers two cases:  If every neighbourhood contains much fewer than ${\Delta \choose 2}$ edges, we can apply a simple probabilistic argument.  Otherwise we have a vertex $v$ 
whose neighbourhood contains almost ${\Delta \choose 2}$ edges.  The fact that $\omega \leq \frac 23(\Delta+1)$ tells us that there is a large antimatching in $N(v)$, and since there are few edges between $N(v)$ and $G-v$, we can take an optimal colouring of $G-N(v)-v$ and extend it to a colouring of $G$ in which many pairs of the antimatching are monochromatic, which is enough to contradict the minimality of $G$.

\section{Dealing with sparse neighbourhoods}

Theorem 10.5 in \cite{molloyrbook}, which is a straightforward application of Talagrand's Inequality, gives us a bound on the chromatic number when no neighbourhood contains almost ${\Delta \choose 2}$ edges:

\begin{theorem}\label{thm:B}
There is a $\Delta_0$ such that for any graph with maximum degree $\Delta>\Delta_0$ and for any $B>\Delta(\log \Delta)^3$, if no $N(v)$ contains more than ${\Delta \choose 2}-B$ edges then $\chi(G)\leq (\Delta+1)-\frac{B}{e^6\Delta}$.
\end{theorem}

We restate this theorem as follows:

\begin{corollary}\label{cor:B2}
There is a $\Delta_0$ such that for any graph with maximum degree at most $\Delta>\Delta_0$ and for any 
$\alpha> 2(\log \Delta)^3 / (\Delta-1)$, if no $N(v)$ contains more than $(1-\alpha){\Delta \choose 2}$ edges then
$$\chi(G)\leq (\Delta+1)-\frac{\alpha(\Delta-1)}{2e^6} \leq \left(1-\frac{\alpha}{2e^6}\right)(\Delta+1).$$
\end{corollary}

This is all we need for the case in which no neighbourhood contains almost ${\Delta \choose 2}$ edges.

\section{Dealing with dense neighbourhoods}

We need the following theorem to extend a colouring when we have a dense neighbourhood.

\begin{theorem}\label{thm:happy}
Let $\alpha$ be any positive constant and let $\epsilon$ be any constant satisfying $0<\epsilon < \frac 16 - 2\sqrt \alpha$.  Let $G$ be a graph with $\omega \leq \frac 23 (\Delta+1)$ and let $v$ be a vertex whose neighbourhood contains more than $(1-\alpha){\Delta \choose 2}$ edges.  Then

$$\chi(G) \leq \max\{\chi(G-v), (1-\epsilon)(\Delta+1)\}.$$
\end{theorem}

This immediately implies:

\begin{corollary}\label{cor:happy}
Let $\rho$ be a positive constant satisfying $\rho \leq \frac{1}{160}$, let $G$ be a graph with maximum degree at most $\Delta$, $\omega \leq \frac 23 (\Delta+1)$ and let $v$ be a vertex whose neighbourhood contains at least $(1-\rho){\Delta \choose 2}$ edges.  Then
$$\chi(G) \leq \max\{\chi(G-v), (1-\rho)(\Delta+1)\}.$$
\end{corollary}

Before we prove Theorem \ref{thm:happy} we need to lay out one more simple fact:

\begin{lemma}\label{lem:hall}
Every graph $G$ contains an antimatching of size $\lfloor\frac 12( n-\omega(G))\rfloor$.
\end{lemma}
\begin{proof}
Let $M$ be a maximum antimatching; there are $n-2|M|$ vertices outside $M$, and these vertices must form a clique.  Thus $\omega(G)\geq n-2|M|$; the result follows.
\end{proof}

\begin{proof}[Proof of Theorem \ref{thm:happy}]
We may assume that $d(v)=\Delta$ since if this is not the case we can hang pendant vertices from $v$, and we may assume $\alpha < \frac 1{144}$, otherwise no valid value of $\epsilon$ exists.  Our approach is as follows.  We first partition $\tN(v)$ into sets $D_1$, $D_2$, and $D_3$ such that $D_1$ and $D_2$ are small, each $u\in D_2$ has few neighbours outside $D_2\cup D_3$, and each $u\in D_3$ has very few neighbours outside $D_3$.  In particular, $v\in D_3$.  Then, using at most $\max\{\chi(G-v), (1-\epsilon)(\Delta+1)\}$ colours, we first colour $G-(D_2\cup D_3)$, then greedily extend the colouring to $D_2$.  Finally, we exploit the existence of a large antimatching in $G|D_3$ and extend the colouring to $D_3$ using an elementary result on list colourings.

It is straightforward to confirm that there are at most $\alpha(\Delta^2-\Delta)$ edges between $G-\tN(v)$ and $\tN(v)$.  We set $c_1 = \frac 12$ and $c_2 = \sqrt \alpha $.  We partition $N(v)$ into $D_1$, $D_2$, and $D_3$ as follows:
\begin{eqnarray*}
D_1 &=& \{ \ u\in N(v) \ \mid \ u \textup{ has more than } c_1(\Delta+1) \textup{ neighbours outside } \tN(v) \ \}\\
D_2 &=& \{ \ u\in N(v)\setminus D_1 \ \mid \ u \textup{ has more than } c_2(\Delta+1) \textup{ neighbours outside } \tN(v)\setminus D_1 \ \}\\
D_3 &=& \tN(v)\setminus(D_1\cup D_2) 
\end{eqnarray*}
Let $\beta_1$ denote $|D_1|/(\Delta+1)$ and let $\beta_2$ denote $|D_2|/(\Delta+1)$.  Thus $|D_3|=(1-\beta_1-\beta_2)(\Delta+1)$.  Since there fewer than $\alpha\Delta^2$ edges between $\tN(v)$ and $G-\tN(v)$, we can see that $|D_1| <  2\alpha\Delta < \frac 13\sqrt \alpha(\Delta+1)$.  Note that every vertex in $D_1$ has more neighbours outside $\tN(v)$ than in $\tN(v)$, so there are fewer than $\alpha\Delta^2$ edges between $D_2\cup D_3$ and $G-(D_2\cup D_3)$.  Thus $|D_2| < \sqrt \alpha(\Delta+1)$.  Therefore $\beta_1 < 2\alpha < \frac 16\sqrt\alpha$ and $\beta_2 < c_2 = \sqrt\alpha$.  By the first of these two facts, we can see that $v$ is in $D_3$.

Now let $k$ denote $\lfloor(1-\epsilon)(\Delta+1)\rfloor$, let $k'$ denote $\max\{k,\chi(G-v)\}$, and take a $k'$-colouring of $G-(D_2\cup D_3)$.  We greedily extend this to a $k'$-colouring of $G-D_3$.  To see that this is possible, note that while extending, every vertex in $D_2$ has at most $|D_1|+|D_2|+c_1(\Delta+1)-1 = (\beta_1+\beta_2+c_1)(\Delta+1)-1$ coloured neighbours, so each vertex has at least $k-(\beta_1+\beta_2+c_1)(\Delta+1)+1 > (\frac 12- \epsilon - \frac 76\sqrt\alpha)(\Delta+1) > 0$ available colours, so we can indeed extend to all vertices of $D_2$ greedily.

Extending the partial colouring to $D_3$ takes a little more finesse.  By assumption, $\omega(G|D_3)\leq\frac 23(\Delta+1)$.  Let $M$ be a maximum antimatching in $G|D_3$.  We now define the graph $G_3$ as a clique of size $|D_3|$ minus $|M|$ vertex-disjoint edges.  Note that $G|D_3$ is a subgraph of $G_3$.  Lemma \ref{lem:hall} along with a classical result of Erd\H{o}s, Rubin, and Taylor on list colourings of complete multipartite graphs with parts of size $\leq 2$ \cite{erdosrt80} tells us that $\chi_l(G|D_3) \leq \chi_l(G_3)=\chi(G_3)=|D_3|-|M|\leq \frac 56(\Delta+1)$ (this can be proven easily using induction and Hall's Theorem).  It follows that if we give each vertex of $D_3$ a list of at least $\frac 56(\Delta+1)$ colours, we can find a colouring of $G|D_3$ such that every vertex gets a colour from its list.

We extend the partial colouring of $G-D_3$ to a colouring of $G$ by assigning each vertex $u$ in $D_3$ a list $\ell_u$ consisting of all colours from $1$ to $k$ not appearing in $N(u)\setminus D_3$.  Each list has size at least $k-(\beta_2+c_2)(\Delta+1) > (1-\epsilon-\beta_2-c_2)(\Delta+1)-1 > (1-\epsilon-2\sqrt\alpha)(\Delta+1) -1 > \frac 56(\Delta+1) -1$.  Since the list sizes are integers, each list has size at least $|D_3|-|M|$.  Therefore we can extend the $k'$-colouring of $G-D_3$ to a $k'$-colouring of $G$.  This completes the proof.
\end{proof}

\section{Putting it together}

We can now prove Theorem \ref{thm:main}.

\begin{proof}[Proof of Theorem \ref{thm:main}]
Take $\Delta_0$ from the statement of Corollary \ref{cor:B2} and set $\epsilon$ as $\min\{\frac 1{\Delta_0}, \frac1{320 e^6}\}$.

Let $G$ be a counterexample on a minimum number of vertices and denote its maximum degree and clique number by $\Delta$ and $\omega$ respectively.  If $\Delta \leq \Delta_0$ then the result is implied by Brooks' Theorem, so we can assume $\Delta > \Delta_0$.  If $\omega < \frac 23 (\Delta+1)$, then Theorem \ref{thm:king} guarantees that we have a maximal stable set $S$ such that $\Delta(G-S)<\Delta$ and $\omega(G-S)<\omega$.  By the minimality of $G$ we have a proper colouring of $G-S$ using $$\lceil (1-\epsilon)(\Delta(G-S)+1) +\epsilon\omega(G-S)\rceil < \lceil (1-\epsilon)(\Delta+1) +\epsilon\omega\rceil$$ colours, to which we can add $S$ as a colour class, giving the desired colouring of $G$.  So $G$ satisfies $\omega \leq \frac 23(\Delta+1)$.

Now $G$ must be vertex-critical, must satisfy $\omega \leq \frac 23 (\Delta+1)$ and $\Delta > \Delta_0$, and must have chromatic number $>(1-\frac{1}{320e^6})(\Delta+1)$.  Thus by Corollary \ref{cor:happy} there is no vertex $v$ such that the neighbourhood of $v$ contains more than $(1-\frac{1}{160}){\Delta \choose 2}$ edges.  The theorem now follows immediately from Corollary \ref{cor:B2}.
\end{proof}

\bibliography{masterbib}

\appendix

\section{Proving the intermediate results}

To support the claim that our new proof is short, we offer proofs of the results that we have used, namely the choosability result of Erd\H{o}s, Rubin and Taylor, Theorem \ref{thm:king}, and Theorem \ref{thm:B}.  We omit proofs of Hall's Theorem, Brooks' Theorem, the Lov\'asz Local Lemma, and Talagrand's Inequality.  We begin with choosability.

\subsection{Chromatic choosability in the complement of a matching}

\begin{theorem}[\cite{erdosrt80}]
If $G$ is a graph obtained from $K_n$ by removing a matching of size $\ell$, then $G$ is $(n-\ell)$-chooseable.
\end{theorem}

\begin{proof}
We proceed by induction on $n$; the basis $n=1$ clearly holds.  Let $G$ be a graph obtained from $K_n$ by removing a matching of size $\ell$, in which every vertex $v$ is assigned a list $L(v)$ of at least $n-\ell$ colours.

If $\ell < n/2$ we may take a universal vertex $v$ in $G$, assign it any colour from its list, and delete this colour from all other lists, proceeding by induction in the obvious way.  Thus we may assume $\ell = n/2$.  Call the vertices of $G$ $u_1,\ldots,u_\ell$ and $v_1,\ldots,v_\ell$ such that for $1\leq i\leq \ell$, $u_i$ is nonadjacent to $v_i$.  By the same argument we used for a universal vertex, we can see that $G$ is $n-\ell$-chooseable if some $u_i$ and $v_i$ have non-disjoint lists, so we may assume that for all $i$, the lists of $u_i$ and $v_i$ are disjoint.

We now construct an auxiliary bipartite graph $H$ with parts $V$ and $V'$ in which $V=V(G)$, $V'$ is the set of colours in some list, and $v\in V$ is adjacent to $v'\in V'$ precisely if $v'\in L(v)$.  It suffices to prove that there is a $V$-saturating matching in $H$.  Observe that for a set $W\subseteq V$, $|\{\cup L(v)\mid v\in W\}| \geq n/2$ if $|W|\geq 1$, and that $|\{\cup L(v)\mid v\in W\}| \geq n$ if $W$ intersects both $\{u_i\}_{i=1}^\ell$ and $\{v_i\}_{i=1}^\ell$, which is always the case if $|W|> n/2$.  Therefore the result follows immediately from Hall's Theorem.
\end{proof}

\subsection{Proof of Theorem \ref{thm:king}}

This subsection is essentially a terse version of \cite{king11} with proofs of two lemmas added.  The main intermediate result is the following extension of Haxell's Theorem \cite{haxell95}, the proof of which we postpone until we have proved Theorem \ref{thm:king}:

\begin{theorem}\label{thm:isr}
For a positive integer $k$, let $G$ be a graph with vertices partitioned into cliques $V_1,\ldots V_r$.  If for every $i$ and every $v\in V_i$, $v$ has at most $\min\{k, |V_i|-k\}$ neighbours outside $V_i$, then $G$ contains a stable set of size $r$.
\end{theorem}

To prove Theorem \ref{thm:king} we must investigate intersections of maximum cliques.  Given a graph $G$ and the set $\fc$ of maximum cliques in $G$, we define the {\em clique graph} $G(\fc)$ as follows.  The vertices of $G(\fc)$ are the cliques of $\fc$, and two vertices of $G(\fc)$ are adjacent if their corresponding cliques in $G$ intersect.  For a connected component $G(\fc_i)$ of $G(\fc)$, let $D_i\subseteq V(G)$ and $F_i\subseteq V(G)$  denote the union and the mutual intersection of the cliques of $\fc_i$ respectively, i.e.\ $D_i=\cup_{C\in \fc_i}C$ and $F_i=\cap_{C\in \fc_i}C$.

The proof uses three intermediate results.  The first, due to Hajnal \cite{hajnal65} (also see \cite{rabern11}), tells us that for each component of $G(\fc)$, $|D_i|+|F_i|$ is large:

\begin{lemma}[Hajnal]\label{lem:hajnal}
Let $G$ be a graph and $\fc = \{C_1,\ldots,C_r\}$ be a nonempty collection of maximum cliques in $G$.  Then
$$ |\cap\fc| + |\cup \fc| \geq 2\omega(G).$$
\end{lemma}

\begin{proof}
Fix $G$ and proceed by induction on $r$, with the base case $r=1$ being clear.  Fix $r>1$ and let $\fc'$ denote $\{C_1,\ldots,C_{r-1}\}$.

Let $A$ denote $\cap\fc' \setminus \cap\fc = \cap\fc' \setminus C_r$, and let $B$ denote $\cup\fc \setminus \cup\fc' = \cup \fc' \cup C_r$.  The result holds by induction if $|A|\leq |B|$ since $|\cap\fc| + |\cup \fc| = |\cap\fc'|+|\cup\fc'|+(B-A)$, so assume $|A|>|B|$.  But observe that $(C_r\setminus B )\cup A$ is a clique that is larger than $C_r$, a contradiction.  The lemma follows.
\end{proof}

The second is due to Kostochka \cite{kostochka80}. It tells us that if $\omega(G)$ is sufficiently close to $\Delta(G)+1$, then $|F_i|$ is large:

\begin{lemma}[Kostochka]\label{lem:kostochka}
Let $G$ be a graph with $\omega(G)>\frac 23(\Delta(G)+1)$ and let $\fc$ be the set of maximum cliques in $G$.  Then for each connected component $G(\fc_i)$ of $G(\fc)$,
$$\left| \cap\fc\right| \geq 2\omega(G)-(\Delta(G)+1).$$
\end{lemma}

\begin{proof}
We proceed by induction on $|\fc|$, with the base case $|\fc|=1$ being a consequence of Lemma \ref{lem:hajnal}.  We may consider the components of $G(\fc_i)$ individually, so assume $G(\fc_i)$ is connected.  Let the set $\fc$ of maximum cliques be $\{C_1,\ldots,C_r\}$, and let $\fc'$ denote $\{C_1,\ldots,C_{r-1}\}$.  Observe that it suffices to deal with the case in which every vertex of $G$ is in $\cup\fc$.

Since $\omega(G)>\frac 23(\Delta(G)+1)$ any two intersecting maximum cliques intersect in more that $\omega/2$ vertices, so the relation of two maximum cliques intersecting is transitive and therefore an equivalence relation.  Furthermore, since $|\cap\fc'|>\omega/2$, $\cap\fc$ is nonempty and therefore $|V(G)|=|\cup\fc|=\Delta+1$.  Thus the result follows from Lemma \ref{lem:hajnal}.
\end{proof}

The third intermediate result is Theorem \ref{thm:isr}.  Combining them to prove Theorem \ref{thm:main} is a simple matter.

\begin{proof}[Proof of Theorem \ref{thm:main}]
Let $\fc$ be the set of maximum cliques of $G$, and let the connected components of $G(\fc)$ be $G(\fc_1),\ldots,G(\fc_r)$.  It suffices to prove the existence of a stable set $S$ in $G$ intersecting each clique $F_i$.  

Lemma \ref{lem:kostochka} tells us that $|F_i|>\frac 13(\Delta(G)+1)$.  Consider a vertex $v\in F_i$, noting that $v$ is universal in $G[D_i]$.  By Lemma \ref{lem:hajnal}, we know that $|F_i|+|D_i|>\frac 43(\Delta(G)+1)$.  Therefore $\Delta(G)+1-|D_i|<|F_i|-\frac 13(\Delta(G)+1)$, so $v$ has fewer than $|F_i|-\frac 13(\Delta(G)+1)$ neighbours in $\cup_{j\neq i}F_i$.  Furthermore $v$ certainly has fewer than $\frac 13(\Delta(G)+1)$ neighbours in $\cup_{j\neq i}F_i$.

Now let $H$ be the subgraph of $G$ induced on $\cup_i F_i$, and let $k=\frac 13(\Delta(G)+1)$.  Clearly the cliques $F_1,\ldots,F_r$ partition $V(H)$.  A vertex $v\in F_i$ has at most $\min\{k,|F_i|-k\}$ neighbours outside $F_i$.  Therefore by Theorem \ref{thm:isr}, $H$ contains a stable set $S$ of size $r$.  This set $S$ intersects each $F_i$, and consequently it intersects every clique in $\fc$, proving the theorem.
\end{proof}

It remains to prove Theorem \ref{thm:isr}.

\subsubsection{Independent transversals with lopsided sets}\label{sec:isr}

Suppose we are given a finite graph whose vertices are partitioned into stable sets $V_1,\ldots, V_r$.  An {\em independent system of representatives} or {\em ISR} of $(V_1,\ldots,V_r)$ is a stable set of size $r$ in $G$ intersecting each $V_i$ exactly once.  A {\em partial ISR}, then, is simply a stable set in $G$ intersecting no $V_i$ more than once.

A {\em totally dominating set} $D$ is a set of vertices such that every vertex of $G$ has a neighbour in $D$, including the vertices of $D$.  Given $J\subseteq [m]$, we use $V_J$ to denote $(V_i \mid i\in J)$.  Given $X\subseteq V(G)$, we use $I(X)$ to denote the set of partitions intersected by $X$, i.e.\ $I(X) = \{ i\in [r] \mid V_i\cap X \neq \emptyset \}$.  For an induced subgraph $H$ of $G$, we implicitly consider $H$ to inherit the partitioning of $G$.

To prove our lopsided existence condition for ISRs, we use a slight strengthening of a lemma of Aharoni, Berger, and Ziv \cite{aharonibz07}.

\begin{lemma}\label{lem:aharoni}
Let $x_1$ be a vertex in $V_r$, and suppose $G[V_{[r-1]}]$ has an ISR.  Suppose there is no $J\subseteq [r-1]$ and $D\subseteq V_J\cup \{x_1\}$ totally dominating $V_{J}\cup \{x_1\}$ with the following properties:
\begin{enumerate*}
\item $D$ is the union of disjoint stable sets $X$ and $Y$.
\item $Y$ is a (not necessarily proper) partial ISR for $V_J$.  Thus $|Y|\leq |J|$.
\item Every vertex in $Y$ has exactly one neighbour in $X$.  Thus $|X|\leq |Y|$.
\item $X$ contains $x_1$.
\end{enumerate*}
Then $G$ has an ISR containing $x_1$.
\end{lemma}

\begin{proof}
Let $G$ be a minimum counterexample; we can assume $G = G[V_{[r-1]}\cup \{x_1\}]$.  Furthermore, $r>1$ otherwise the lemma is trivial.  Let $R_1$ be an ISR of $G[V_{[r-1]}]$ chosen such that the set $Y'_1 = Y_1 = R_1 \cap N(x_1)$ has minimum size.  We know that $R_1$ exists because $G[V_{[r-1]}]$ has at least one ISR, and we know that $Y'_1$ is nonempty because $G$ does not have an ISR.  Now let $X_1 = \{x_1\}$ and let $D_1=X_1\cup Y_1$.

We now construct an infinite sequence of partial ISRs $Y_1 \subset Y_2 \subset \ldots$, which contradicts the fact that $G$ is finite.  Let $i > 1$, and suppose we have sets $\{R_j, Y_j, X_j \mid 1\leq j <i\}$ such that:
\begin{itemize}
\item $X_j$ is a stable set consisting of distinct vertices $\{x_1, \ldots, x_j\}$.  For $j>1$, $x_j$ is a vertex in $G[V_{I(Y_{j-1})}]$ with no neighbour in $X_{j-1}\cup Y_{j-1}$.

\item $R_j$ is an ISR of $G[V_{[r-1]}]$ such that for every $1 \leq \ell < j$, $R_j \cap N(X_\ell) = Y_\ell$.  Subject to that, $R_j$ is chosen so that $Y'_j = R_j \cap N(x_j)$ is minimum.  For $1 \leq j < i$, $Y'_j$ is nonempty.

\item $Y_j = \cup_{i=1}^j Y'_j$.
\end{itemize}

To find $x_i$, $Y'_i$, and $R_i$, we proceed as follows.
\begin{enumerate}
\item Let $x_i$ be any vertex in $G[V_{I(Y_{i-1})}]$ with no neighbour in $X_{i-1}\cup Y_{i-1}$.  We know that $x_i$ exists, otherwise the set $D_{i-1}=X_{i-1}\cup Y_{i-1}$ would be a total dominating set for $G[V_{I(Y_{i-1})} \cup \{x_1\}]$, contradicting the fact that $G$ is a counterexample.

\item Let $R_i$ be an ISR of $G[V_{[r-1]}]$ chosen so that for all $1\leq j < i$, $R_i \cap N(x_j) = R_j \cap N(x_j) = Y'_j$.  Subject to that, choose $R_i$ so that $Y'_i = R_i\cap N(x_i)$ is minimum.  We know that $R_i$ exists because $R_{i-1}$ is a possible candidate for the ISR.

\item It remains to show that $Y'_i$ is nonempty, i.e.\ that $Y_i \neq Y_{i-1}$.  Suppose $Y'_i = \emptyset$.  We will show that this contradicts our choice of $R_j$ for the unique $j<i$ such that $x_i \in V_{I(Y'_j)}$.  Let $y$ be the unique vertex in $R_i \cap V_{I(x_i)}$.  Construct $R'_j$ from $R_i$ by removing $y$ and inserting $x_i$.  Now for every $\ell$ such that $1\leq \ell < j$, $R'_j \cap N(x_\ell) = Y'_\ell = R_j \cap N(x_\ell)$.  For $j$, $R'_j \cap N(x_j) = (R_j \cap N(x_j))\setminus \{ y\}$, a contradiction.  Thus $Y'_i$ is nonempty.

\item Set $X_i = X_{i-1}\cup \{x_i\}$ and $Y_i = Y_{i-1}\cup Y'_i$.
\end{enumerate}
This choice of $X_i$, $R_i$, and $Y_i$ sets up the conditions so that we can repeat our argument indefinitely for increasing $i$, a contradiction since $G$ is finite.
\end{proof}

Theorem \ref{thm:isr} follows immediately from the following consequence of the previous lemma:

\begin{theorem}\label{thm:isr2}
Let $k$ be a positive integer and let $G$ be a graph partitioned into stable sets $(V_1,\ldots,V_r)$.  If for each $i\in [r]$, each vertex in $V_i$ has degree at most $\min\{k,|V_i|-k\}$, then for any vertex $v$, $G$ has an ISR containing $v$.
\end{theorem}

\begin{proof}
Suppose $G$ is a minimum counterexample for a given value of $k$.  Clearly we can assume each $V_i$ has size greater than $k$, and that $G[V_J]$ has an ISR for all $J\subset [r]$.  Take $v$ such that $G$ does not have an ISR containing $v$; we can assume $v\in V_r$.  By Lemma \ref{lem:aharoni}, there is some $J \subseteq [r-1]$ and a set $D\subseteq V_J\cup \{v\}$ totally dominating $V_J\cup \{v\}$ such that (i) $D$ is the union of disjoint stable sets $X$ and $Y$, (ii) $Y$ is a partial ISR of $V_J$, (iii) $|X|\leq |Y|\leq |J|$, and (iv) $v\in X$.

Since $D$ totally dominates $V_J\cup \{v\}$, the sum of degrees of vertices in $D$ must be greater than the number of vertices in $V_J$.  That is, $\sum_{v\in D}d(v) > \sum_{i\in J}|V_i|$.  Clearly $\sum_{v\in X}d(v) \leq k\cdot |J|$ and $\sum_{v\in Y}d(v) \leq \sum_{i\in J}(|V_i|-k)$, so $\sum_{v\in D}d(v) \leq \sum_{i\in J}|V_i|$.  Thus $D$ cannot totally dominate $V_J\cup \{v\}$, giving us the contradiction that proves the theorem.
\end{proof}

\subsection{Proof of Theorem \ref{thm:B}}

Here we prove Theorem 10.5 from \cite{molloyrbook}, which is an straightforward application of Talagrand's Inequality and the Lov\'asz Local Lemma.

\begin{lemma}[Lov\'asz Local Lemma]
Let $\mathcal A$ be a set of events in a probability space and take $p\in \mathbb R$ and $d \in \mathbb Z$ such that for every $A\in \mathcal A$,
\begin{itemize*}
\item $\Pr(A) \leq p$ and
\item $A$ is independent of all but at most $d$ other events in $\mathcal A$.
\end{itemize*}
Then if $4pd\leq 1$, with nonzero probability no event in $\mathcal A$ occurs.
\end{lemma}

\begin{theorem}[Talagrand's Inequality]
Let $X$ be a non-negative random variable, not identically 0, which is determined by $n$ independent trials $T_1,\ldots,T_n$, and satisfying the following for some $c,r>0$:
\begin{itemize*}
\item changing the outcome of any one trial changes the value of $X$ by at most $c$, and
\item for any $s$, if $X\geq s$ then there is a set of at most $rs$ trials whose outcomes certify that $X\geq s$,
\end{itemize*}
then for any $0\leq t\leq E(X)$,
$$
\Pr(|X-\E(X)| > t + 60c\sqrt{r\E(X)}) \leq 4e^{-\frac{t^2}{ 8c^2r\E(X) }}.
$$
\end{theorem}

\begin{proof}[Proof of Theorem \ref{thm:B}]
A simple embedding argument allows us to assume that $G$ is $\Delta$-regular (take two copies of $G$, add an edge between the two copies of each vertex of minimum degree, and repeat as necessary).  Set $C = \lfloor \Delta/2\rfloor$ and {\em assign} every vertex a colour in $\{1,\ldots,C\}$ uniformly at random.  If a vertex $w$ is assigned a colour that appears on some neighbour, we {\em uncolour} $w$ and all its neighbours of the same colour.  Otherwise we say $w$ {\em retains} its colour.

We wish to lower-bound the number of colours retained by at least two neighbours of a given vertex.  To do so we will underestimate this number with the more manageable variable $X_v$, which we define as the number of colours assigned to at least two non-adjacent neighbours of $v$ and retained by all neighbours of $v$ to which they are assigned.

We consider two closely related variables for each vertex, which we may as well introduce now.  Let the variable $AT_v$ (assigned twice) count the number of colours assigned to at least two non-adjacent neighbours of $v$, and let the variable $Del_v$ (deleted) count the number of colours assigned to at least two non-adjacent neighbours of $v$ but removed (i.e.\ uncoloured) from at least one neighbour of $v$.  Note that $X_v = AT_v - Del_v$.

For all $v\in V(G)$ let $A_v$ be the event that $X_v < \frac{B}{e^6\Delta}$.  To prove the theorem it suffices to prove that with nonzero probability, $A_v$ holds for no vertex $v$.  To see this, note that if we have a colouring using $C$ colours in which every vertex has at least $\frac{B}{e^6\Delta}$ repeated colours in its neighbourhood, we can complete a $\Delta+1-\frac{B}{e^6\Delta}$ colouring of $G$ as follows: First we extend each of the $C$ colour classes such that if a vertex $v$ is uncoloured, all $C$ colours appear on its neighbourhood; we can do this greedily one colour at a time.  We then delete these $C$ colour classes, giving us a graph of maximum degree at most $\Delta-C-\frac{B}{e^6\Delta}$, which we can then colour greedily using $\Delta+1-C-\frac{B}{e^6\Delta}$ colours.

We will prove in three separate lemmas that:
\begin{itemize*}
\item $\E(X_v) \geq \frac{1.99B}{e^6\Delta}$ (Lemma \ref{lem:expectation}).
\item $\E(X_v) \leq \E(AT_v) \leq \frac{3B}{\Delta} \leq \frac{e^6}{3}\E(X_v)$ (Lemma \ref{lem:expectation2}).
\item $\Pr\left(|X_v - \E(X_v)| > \log \Delta \sqrt{\E(X_v)}\right) < \frac{1}{4\Delta^5}$ (Lemma \ref{lem:concentration}).
\end{itemize*}
Now for $B \geq \Delta(\log\Delta)^3$ with sufficiently large $\Delta$, we have
$$\tfrac{1.99B}{e^6\Delta} - \log\Delta\sqrt{\E(X_v)}\  \geq \ \tfrac{1.99B}{e^6\Delta} - \log\Delta\sqrt{\tfrac{3B}{\Delta}} \ > \ \tfrac{B}{e^6\Delta}.$$
Therefore for any $v$ we have $\Pr(A_v) < 1/(4\Delta^5)$.  Since $A_v$ only depends on the colours assigned to vertices at distance at most two from $v$, an event $A_u$ is independent from $A_v$ unless $u$ is at distance at most four from $v$; there are at most $\Delta^4$ such events.  Thus setting $p=1/(4\Delta^5)$ and $d=\Delta^4$, the result follows from the Local Lemma.
\end{proof}

\begin{lemma}\label{lem:expectation}
$\E(X_v) \geq \frac{1.99B}{e^6\Delta}$.
\end{lemma}
\begin{proof}
For every vertex $v$ we define $X'_v$ to be the number of colours assigned to exactly two non-adjacent neighbours of $v$ and retained by both.  Note that $X_v \geq X'_v$.

Two vertices $u,w\in N(v)$ will both retain the colour $\alpha$ if $\alpha$ is assigned to both $u$ and $v$ but no vertex in $S= N(v)\cup N(u)\cup N(w)-u-w$.  Because $|S|\leq 3\Delta-3\leq 6C$, for any colour $\alpha$ the probability that this occurs is at least $(\tfrac 1C)^{2}(1-\tfrac 1C)^{6C}$.  There are $C$ choices for $\alpha$ and at least $B$ choices for $\{u,v\}$.  Therefore by Linearity of Expectation for sufficiently large $C=\lfloor \Delta/2\rfloor$, \mbox{we have}
$$
\E(X'_v)\geq CB\left( \frac 1 C \right)^2 \left(1- \frac 1 C \right) ^{6C} =  \frac BC \left(1- \frac 1 C \right) ^{6C}  \geq \frac{2B}{\Delta} \left(1- \frac 1 C \right) ^{6C}  \geq \frac{1.99B}{e^6\Delta}.\qedhere
$$
\end{proof}

\begin{lemma}\label{lem:expectation2}
$\E(X_v) \leq \E(AT_v)\leq \frac{3B}{\Delta} \leq \frac{e^6}{3}\E(X_v)$.
\end{lemma}
\begin{proof}
The probability of a colour $\alpha$ being assigned to at least two nonadjacent neighbours of $v$ is at most $\frac{B}{C^2}$, therefore $\E(AT_v) \leq \frac{B}{C}\leq \frac{3B}{\Delta} \leq \frac{e^6}{3}\E(X_v)$, the last inequality coming from Lemma \ref{lem:expectation}.
\end{proof}

\begin{lemma}\label{lem:concentration}
$\Pr\left(|X_v - \E(X_v)| > \log \Delta \sqrt{\E(X_v)}\right) < \frac{1}{4\Delta^5}$.
\end{lemma}

\begin{proof}
To prove the lemma it suffices to prove that the following concentration bounds hold for $t>\sqrt{\Delta\log\Delta}$:
\begin{itemize}
\item{\em Claim 1: }$\Pr(|AT_v - \E(AT_v)| > t) < 4e^{-\frac{t^2}{100\E(AT_v)}}$.
\item{\em Claim 2: }$\Pr(|Del_v - \E(Del_v)| > t) < 4e^{-\frac{t^2}{100\E(AT_v)}}$.
\end{itemize}

To see that these claims imply the lemma, we first observe that $\E(X_v)=\E(AT_v)-\E(Del_v)$.  Therefore if $|X_v - \E(X_v)| > \log \Delta \sqrt{\E(X_v)}$, setting $t = \frac 12 \log \Delta \sqrt{\E(X_v)}  > \sqrt{\Delta\log\Delta}$, either $|AT_v - \E(AT_v)| > t$ or $|Del_v - \E(Del_v)| > t$.  Thus by the claims, and noting that 
$$t^2 \geq \tfrac 14(\log \Delta)^2 ({\tfrac{3}{e^6}\E(AT_v)}) > \tfrac{1}{2e^6}(\log \Delta)^2 \E(AT_v) ,
$$ 
the probability of this happening is at most
$$
8e^{-\frac{t^2}{100\E(AT_v)}} <
8e^{-\frac{(\log\Delta)^2\E(AT_v)}{200e^6\E(AT_v)}} =
8e^{-\frac{(\log\Delta)^2}{200e^6}} <
\frac{1}{4\Delta^5}.
$$

We now prove Claim 1.  The value of $AT_v$ only depends on the colours assigned to $N(v)$, and changing any of these assignments can affect $AT_v$ by at most 2.  If $AT_v\geq s$ then there is a set of at most $2s$ assignments that certify this.  Therefore Talagrand's Inequality with $c=2$ and $r=2$ gives us:
$$
\Pr(|AT_v - \E(AT_v)| > t) < 4e^{- \frac{
\left(t-120\sqrt{ 2\E(AT_v)  } \right)^2
}{
64\E(AT_v)
}    }
<
4e^{-\frac{t^2}{100\E(AT_v)}},
$$
the latter inequality following from the fact that $t\geq \sqrt{\Delta\log\Delta} \geq \sqrt{\E(AT_v)\log\Delta}$.

We now prove Claim 2 in the same way.  The value of $Del_v$ depends on at most $\Delta^2+1$ colour assignments.  As with $AT_v$, changing a colour assignment changes $Del_v$ by at most 2.  If $Del_v\geq s$, this can be certified by a set of at most $3s$ assignments (the two nonadjacent vertices in $N(v)$ and a third vertex adjacent to one of the first two).  Therefore we can apply Talagrand's Inequality with $c=2$ and $r=3$, which gives us:
$$
\Pr(|Del_v - \E(Del_v)| > t) < 4e^{- \frac{
\left(t-120\sqrt{ 3\E(Del_v)  } \right)^2
}{
96\E(Del_v)
}    }
<
4e^{-\frac{t^2}{100\E(AT_v)}},
$$
since $\E(AT_v)\geq \E(Del_v)$ and $t\geq \sqrt{\Delta\log\Delta} \geq \sqrt{\E(AT_v)\log\Delta}$.
\end{proof}

\end{document}